\RequirePackage{fix-cm}
\documentclass[11pt,twoside,a4paper]{article}
\usepackage{booktabs}

\usepackage{enumitem}
\usepackage{color}
\usepackage{graphicx}
\usepackage{amsmath}
\usepackage{amsfonts}
\usepackage{amssymb }
\usepackage{float}
\usepackage{multicol}
\usepackage{amsthm}
\usepackage{array}
\usepackage{tabularx}

\usepackage{latexsym}

\newtheorem{theorem}{Theorem}
\newtheorem{mytheorem}{Theorem}

\newtheorem{lemma}[mytheorem]{Lemma}
\newtheorem{remark}[mytheorem]{Remark}

\begin{document}

\title{Coexisting Hidden and self-excited attractors in an economic system of integer or fractional order}

\author{Marius-F. Danca\\
Romanian Institute of Science and Technology, \\Cluj-Napoca, Romania\\
email: danca@rist.ro}

\maketitle

\begin{abstract}In this paper the dynamics of an economic system with foreign financing, of integer or fractional order, are analyzed. The symmetry of the system determines the existence of two pairs of coexisting attractors. The integer-order version of the system proves to have several combinations of coexisting hidden attractors with self-excited attractors. Because one of the system variables represents the foreign capital inflow, the presence of hidden attractors could be of a real interest in economic models. The fractional-order variant presents another interesting coexistence of attractors in the fractional order space.
\end{abstract}

\textbf{Keywords }Hidden attractor; self-excited attractor; coexisting attractors; saddle; economic system

\section{Introduction}

According to Shilnikov criteria, chaos emergence requires at least one unstable equilibrium \cite{patru}. So, if for some values of the key parameter the equilibria are unstable and one tries to generate attractors with initial conditions near these equilibria, one obtains chaotic attractors and could deduce that the system evolves chaotically. On the other side, this conclusion could be incomplete, or even false, since for initial conditions taken from an attraction basin which does not intersect with any small neighborhoods of unstable equilibria, the underlying trajectories could lead to some other attractor that might be a regular motion, but not a chaotic motion. Recently this problem has been managed by introducing the term of hidden attractor (see e.g. \cite{cinci}, one of the first works on this subject, or \cite{noua}). If the attraction basin intersects with any
open neighborhoods of an equilibrium the attractor is called self-excited, otherwise, it is called hidden. Usually, the self-excited attractors are numerically derived from unstable equilibria, while hidden attractors are difficult to be localized, because their attraction basins have no relation with small neighborhoods of any equilibria. Hidden attractors can be found in a nonlinear system with one or more stable equilibria \cite{sase,sase_1,sase_2,sase_3}, with a line of equilibria (infinite equilibria) \cite{sapte}, with both unstable
equilibria and stable equilibria \cite {zece}, with coexistence of various attractors \cite{unspe}, or even without equilibria \cite{opt}.
Note that there is a lack of correlation between hidden attractors and equilibria. Moreover, the precise localization of hidden attractors seems to be an intractable issue and, to our best knowledge, there is still no general analytical way, but only numerically with luck, to find hidden attractors (for the particular case of Chua's circuit, see \cite{paispe}).

To understand better the importance of hidden attractors consider the crash of aircraft YF-22 Boeing in April 1992. The analysis of aircrafts and launchers control systems with saturation
regards the linear stability and also hidden oscillations which might occur \cite{doispe}. The difficulties of rigorous analysis and design of nonlinear control systems with saturation related to this case are presented in \cite{treispe}.

Therefore, the identification of hidden attractors in a given system can help drive the system along the desired attractor, which can be a self-excited attractor or a hidden attractor.

The fractional order (FO) calculus is as old as the integer-order (IO) one, but its application was exclusive in mathematics. In the latter years FO systems where proved to describe, generally more accurately and in compact expressions, the behavior of real dynamical systems, compared to the IO models. This happens taking into account the nonlocal characteristic of ``infinite memory'', i.e. the evolution of the system depends at each moment on the entire previous history. Basic aspects of the theory of FO can be found in \cite{ppp}. While the definition of FO derivative for continuous-time real functions has been formulated by Liouville, Grunwald, Letnikov, and Riemann, in the late 19th century, the first definition of a fractional difference operator was proposed in 1974 \cite{xxx}.
In this paper the FO derivative is considered in the sense of Caputo, because it allows the choosing initial conditions for the IO systems. One of the most utilized numerical methods to integrate continuous systems of FO is the predictor-corrector Adams-Bashforth-Moulton predictor-corrector method for fractional-order differential equations \cite{21}.

There are many continuous but also discrete chaotic systems of IO, such as the supply and demand systems which are among the oldest and simplest economic discrete models, having    complicated dynamics \cite{unu}. It is proved that the economical models exhibit generally unstable steady states and also fluctuations if the income distribution varies sufficiently and if shareholders
save more than workers (see e.g. \cite{doi} and references therein).

In this paper a continuous economic system is considered, and proved numerically that it presents coexisting attractors for both IO and FO variants.

\section{The economic system}

In \cite{doi} the question of whether using and/or un-using of foreign investments can change the qualitative properties is addressed for a growth path of a proposed dynamical economic
system with three endogenous variables and only one non-linear term, modeled by 3-dimensional autonomous differential equations, while in \cite{trei} the system is further analyzed via bifurcations routes including supercritical and subcritical Hopf bifurcation, and generalized Hopf bifurcation as well. It is proved that the existing cycle exhibits period-doubling
bifurcation as a route to chaos.

The economic system considered in this paper, with a foreign financing, has the following form

\begin{equation}\label{1}
\begin{array}
[l]{l}
\dot{x}_1=ax_2+bx_1(c-x_2^2),\\
\dot{x}_2=d(x_1+x_3),\\
\dot{x}_3=    ex_1-fx_2,
\end{array}
\end{equation}
where the variables $x_1$ represents the savings of households, $x_2$ the Gross Domestic Product (GDP), $x_3$ the foreign capital inflow, and the variables $a$ the variation of the marginal propensity to savings, $b$ the ratio of capitalized profit, $c$ the potential GDP, $d=1/\nu$, with $\nu$ the capital/output ratio, $e$ the capital inflow/savings ratio and $f$ the debt refund/output ratio \cite{unu}. As explained in \cite{unu}, if $x_2 <c$, the activities are not constrained and high profits are derived from the sectors where the markets are not yet saturated. The single non-linearity represents the quality of a capitalization of the profits. If $x_2>c$, then the inflation is possible. A lack of new investment opportunities modifies it by savings.

Note that the system can also be viewed as an extension of the van der Pol's system on the plane $(x_1,x_2)$, i.e. when economically the system has no foreign investment ($x_3=0$)
\begin{equation}\label{2}
\begin{array}
[l]{l}
\dot{x}_1=ax_2+bx_1(c-x_2^2),\\
\dot{x}_2=dx_1,
\end{array}
\end{equation}
which, as a two-dimensional autonomous system, cannot display chaotic behavior.

In this paper one considers the 3-dimensional system \eqref{1} with $a$ the bifurcation parameter and $b=0.01$, $c=1$, $d=0.031847$, $e=0.19$, $f=0.25$ (there are other interesting possible parameters choices\footnote{Parameter $d$ has been chosen as in \cite{doi}, with six decimals, because rounding to fewer decimals affects significantly the results.})

\begin{equation}\label{3}
\begin{array}
[l]{l}
\dot{x}_1=ax_2+0.01x_1(1-x_2^2),\\
\dot{x}_2=0.031847(x_1+x_3),\\
\dot{x}_3=    0.19x_1-0.25x_2.
\end{array}
\end{equation}

While in \cite{doi} coexisting stable cycles and equilibria and also the existence of a chaotic attractor are studied analytically, in this paper it is shown numerically that for $a<f$ the system presents richer dynamics, including coexisting hidden and self-excited chaotic attractors and stable cycles. This behavior could be extremely interesting from an economical point of view.

Because the right hand side of the system \eqref{1} is an odd function, the system presents symmetries in the bifurcation planes, and also in the phase space.

The equilibria, which are collinear (see Fig. \ref{fig4}), have the following expression
\[
X_0^*=(0,0,0),~~X_{1,2}^*=\Bigg(\pm\sqrt{\frac{bf^2+a e f}{b e^2}},\pm\sqrt{\frac{bf^2+a e f}{b f^2}}, \mp\sqrt{\frac{bf^2+a e f}{b e^2}} \Bigg)
\]

Note that $X_{1,2}^*$ with the considered 4 decimals, are approximations of the exact equilibria $X_{1,2}^*$.

The stability of equilibria and the occurrence of the Hopf bifurcation for $a\geq 0.25$ are studied in \cite{doi}.

In this paper, the bifurcation parameter $a$ is considered, $0\leq a \leq0.2$, where all equilibria are unstable and the system presents most interesting dynamics.

The following result, presented in \cite{doi,trei}, is slightly improved and re-proved numerically here, to be useful for finding hidden attractors later.

\begin{lemma}\cite{doi,trei} For $a\in[0,0.2]$, equilibria $X_{0,1,2}^*$ are hyperbolic unstable saddles. Equilibrium $X_0^*$ is attracting focus-saddle and $X_{1,2}^*$ repelling focus-saddle.
\end{lemma}

\begin{proof}

The stability is proved numerically by analyzing the eigenvalues of the Jacobi matrix which has the following form
\[
J=\left(
\begin{array}{ccc}
0.001-0.001x_2^2 & a-0.02x_1x_2 & 0\\
0.031847  0 & 0 & 0.031847\\
0.190 & -0.250 & 0%
\end{array}%
\right).
\]

\begin{itemize}
\item[i)]

The characteristic equation $P(\lambda)=0$, $\lambda\in \mathbb{R}$, where $P(\lambda)$ is the characteristic polynomial, for $X_0^*$, has the following form
\[
P(\lambda)=\lambda^3-0.0100\lambda^2+(0.0080-0.0318a)\lambda-0.0060a-0.0001.
\]

Because, for $a>0.008/0.0318\approx 0.25$, the coefficient $0.008 - 0.0318a<0$ and there exists a single change in the sign of the coefficients of $P(\lambda):+,-,-,-$. Therefore, by Descartes' rule, there exists a single real positive zero of $P(\lambda)$, $\lambda_1$. For $a\leq0.2$, the coefficient $0.008 - 0.0318a>0$ and the coefficients sign are: $+,-,+,-$, i.e. three sign changes, which means that the polynomial $P(\lambda)$ has either 3 or 1 positive zeros. In order to understand better the characteristics of $X_0^*$, in Fig. \ref{fig1} (a) the three roots of $P(\lambda)$, $\lambda_{1,2,3}$, are plotted as function of $a$, for $a\in[0,2]$. As can be seen, the real eigenvalue $\lambda_{1}$ is positive for all considered values of $a$ and the other two eigenvalues $\lambda_{2,3}$ are complex with $\Re(\lambda_{2,3})<0$. Therefore, $X_0^*$ is spiral saddle of index 1 (or attracting focus saddle). This means that trajectories will leave this equilibrium along the unstable manifold of dimension 1 (generated by the real positive eigenvalue $\lambda_1$), by spiralling, due to the stable manifold of dimension 2 (generated by $\Re(\lambda_{2,3})<0$).

\item[ii)]
For equilibria $X_{1,2}^*$, the characteristic polynomial is

\[
P(\lambda)=\lambda^3+0.0076\lambda^2+(0.0090-0.0318a)\lambda+0.0121a+0.0002.
\]

The proof of instability of $X_{1,2}^*$, via Hurwitz stability criterion, can be found in \cite{doi}.
To verify it numerically, in Fig. \ref{fig1} (b) the zeros of $P(\lambda)$ are plotted as function of $a$ for $a\in[0,2]$. As can be seen, for all $a\in[0,2]$, there exist two complex eigenvalues with positive real parts and one real negative eigenvalue. Therefore $X_{1,2}^*$ are spiral saddle of index 2 (or repelling focus saddle). Trajectories near $X_{1,2}^*$ are attracted along the stable manifold of dimension 1 (generated by the real negative eigenvalue $\lambda_1$) and then are rejected spiraling on the unstable manifold of dimension 2 (due to the positiveness of the real part of $\lambda_{2,3}$, $\Re(\lambda_{2,3})>0$).
\item[iii)] Because for all equilibria $\Re(\lambda_{2,3})\neq 0$, $a\in[0,0.2]$, $X_{0,1,2}^*$ are hyperbolic.
\end{itemize}
\end{proof}
Fig. \ref{fig1} (c) presents a chaotic attractor for $a=0.001$. As can be seen, before reaching the chaotic attractor, the trajectory connects the saddles $X_0^*$ and $X_1^*$, i.e. it is a heteroclinic connection.

\section{Hidden and self-exited attractors}

Hereafter attractors are denoted as follows:
\begin{itemize}
\item[] Self-Excited Cycles: $SEC$;
\item[] Self-Excited Chaotic Attractors: $SECH$;
\item[] Hidden Chaotic Attractors: $H$;
\item[] Hidden Cycles: $HC$.
\end{itemize}

The numerical method utilized for the integration of the system \eqref{3} of IO is the Matlab variant of RK method, $ode45$.

The numerical trajectories for different attractors are represented overplotted in the phase space, as time series and the local Lyapunov exponents are determined. Generally, transients are removed.
Also, excluding the Perron effect (when a negative largest Lyapunov exponent does not, in general, indicate stability, and also when a positive largest Lyapunov exponent does not, in general, indicate chaos \cite{per,per2}), the positiveness of the maximal Lyapunov exponent (MLE) is considered as an evidence for chaos. In this paper the zero MLE is considered as having a maximum value of order $10^{-5}-10^{-4}$. If $MLE$ is zero, then, after transients removed, the trajectory is considered as regular (stable cycle), while if $MLE>0$ the trajectory is considered chaotic.

The bifurcation diagram of the maximum component $x_1$, for $a\in[0,0.2)$, is presented in Fig. \ref{fig2} (a).

As can be seen, for all values of $a$, the underlying attractors are symmetric. Due to the symmetry, the coexistence of attractors features this system. Thus, depending on the initial condition $x_0$, there exists one or even two pairs of symmetric attractors (blue-red and black-yellow in the bifurcation diagram).
Beside the standard period-doubling cascade which leads to chaos, the zoomed region in Fig. \ref{fig2} (b) reveals an interesting coexisting window for $a\in[a_1,a_2]$, with $a_1=0.0485$ and $a_2=0.0524$, where exterior and interior attractors crises, thin periodic windows interleaved between large coexisting chaotic-periodic windows, and also several miniature bifurcation diagrams within a large chaotic window can be seen. Moreover, at $a=a_2$, there exist both exterior and interior crisis.

However the most interesting characteristic of the window presented in Fig. \ref{fig2} (b) represents the coexistence of self-excited chaotic attractors with hidden chaotic attractors.

Because there exists no general analytical way to find the attraction basin of hidden attractors of a given nonlinear system, in the case when the considered system admits unstable equilibria, an acceptable way is to search randomly initial conditions in neighborhoods of all these equilibria. This can be done either in planar (usually horizontal) sections through each equilibrium or/and analyzing three-dimensional (usually spherical) neighborhoods centered on equilibria.

For the considered system \eqref{3}, one considers lattices $B_{0,1}=[-5,5]\times [-5,5]$ of $400\times 400$ points $(x_1,x_2)$ and zoomed regions, as horizontal sections through $X_0^*$ (with $x_3=0$) and $X_1^*$ (with $x_3$ the third coordinate of $X_1^*$), respectively, and also spherical neighborhoods of unstable equilibria.

Because of the symmetry, the case of the equilibrium $X_2^*$ is similar to $X_1^*$. Therefore only $X_0^*$ and $X_1^*$ are analyzed.

\subsection{Coexisting self-excited chaotic attractors and hidden chaotic attractors}\label{self-hid}

\underline{$a=0.052$}. For this value, $X_1^*(2.9280,2.2253,-2.9280)$ and the system presents two pairs of different chaotic attractors,  $SECH_1$, $H_1$ and $SECH_2$, $H_2$ respectively (see Fig. \ref{fig3} (a) for $H_1$ and $H_2$ and Fig. \ref{fig3} (b) for both pairs of hidden attractors and self-excited attractors). The extensive numerical experiments revealed that initial conditions within small neighborhoods of unstable equilibria $X_0^*$ and $X_{1,2}^*$ launch trajectories to \emph{self-excited chaotic attractors} $SECH_{1,2}$ (blue and light brown). Initial conditions outside the attraction basins of the self-excited attractors $SECH_{1,2}$ lead to the other two chaotic attractors, $H_{1,2}$ (red and black), which are \emph{hidden chaotic attractors}. To verify that, one can explore planar neighborhoods of each equilibria in lattices $B_0$, containing $X_0^*$ and defined by $x_3=0$, and $B_1$, containing $X_1^*$ and defined by $x_3=-2.928$ (Fig. \ref{fig4} (a)). All points in $B_{0,1}$ are considered as initial conditions, to see where the underlying trajectories go. As the numerical analysis shows, the attraction basins of $H_1$ and $H_2$ on $B_0$ do not contain (intersect) the equilibrium $X_0^*$ (see the black and red plot in Figs. \ref{fig4} (b), (c)). Similarly, these attraction basins do not intersect the equilibrium $X_1^*$ on $B_1$ (Figs. \ref{fig4} (d), (e)). Therefore $H_{1,2}$ are hidden attractors.

Note that $X_0^*$ is on the separatrix between the attraction basins of $SECH_1$ and $SECH_2$ and, therefore, every (no matter how small) neighborhoods around $X_0^*$ (Fig. \ref{fig4} (c)) share initial conditions toward $SECH_1$ or $SECH_2$. This situation is typical for this system and also for saddles generally. Therefore, $SECH_{1,2}$ are self-excited.

In addition to planar numerical analysis, a three-dimensional test around $X_0^*$ and $X_1^*$ is done by analyzing initial conditions within a sphere surrounding these equilibria. For clarity, 50 random points as initial conditions are chosen. Figs. \ref{fig5} (a), (b) presents the case of equilibrium $X_0^*$ with the neighborhood $V_{X_0^*}$. As can be seen, due to the fact that $X_0^*$ is placed on the separatrix, the neighborhood $V_{X_0^*}$ shares initial conditions for attractors $SECH_1$ and $SECH_2$.

In Figs. \ref{fig5} (c), (d), within a spherical neighborhood $V_{X_1^*}$, 50 initial conditions are all leading to $SECH_1$.

Therefore, because the chaotic attractor $H_{1}^*$ is not connected with any equilibria $X_0^*$ and $X_1^*$, it can be considered hidden.

The self-excited attractors $SECH_{1,2}$ have $MLE= 0.0025$, while for the attractors $H_{1,2}$, $MLE=0.0051$ and, therefore, they are chaotic.

\subsection{Coexisting self-excited cycles and self-excited chaos}

\underline{$a=0.05$}. Within the window $a\in[a_1,0.0508)$ (Fig. \ref{fig2} (b)), as the bifurcation diagram shows, there exists a pair of symmetric stable cycles (red and black in Fig. \ref{fig6}) coexisting with a pair of symmetric chaotic attractors (blue and light-brown). Similar numerical approach to the one discussed in Subsection \ref{self-hid}, shows that for $a=0.05$ and $X_1^*(2.8828,2.1909,-2.8828)$, there exists a pair of stable cycles, $SEC_1$ and $SEC_2$ (red and black plot, Fig. \ref{fig6} (a)) and of chaotic attractors $SECH_1$ and $SECH_2$ (blue and light-brown).

Consider the lattices $B_0$ (Figs. \ref{fig6} (b) and (c)) and $B_1$ (Figs. \ref{fig6} (d) and (e)) in the planar sections through $X_0^*$ and $X_1^*$. As for the case $a=0.052$, $X_0^*$ belongs to a separatrix between attraction basins of $SEC_1$ and $SEC_2$ (Fig. \ref{fig6} (c)) and, therefore, $SEC_{1,2}$ are \emph{self-exited attractors}. The zoom in Fig. \ref{fig6} (e) shows the fact that for this value of $a$, $X_1^*$ also belongs to the separatrix of attraction basins of $SEC_1$ and $SECH_1$ (red and blue respectyively, Fig. \ref{fig6} (d), (e)), confers the status of self-excited characteristic of these attractors.

Even by considering that attraction basins of $SECH_{1,2}$ do not intersect $X_0^*$ (blue and light-brown in Fig. \ref{fig6} (c)), they can be hidden, because of the position of $X_1^*$, (red in Fig. \ref{fig6} (e)) $SECH_{1,2}$ are actually not hidden.

Note that with an extremely small perturbation of $a$, $\Delta a=2e-3$, from the previous case of $a=0.052$, the attraction basins suffered significant modifications around $X_1^*$, so that the hidden characteristic vanished.

The attractors $SEC_{1,2}$ have $MLE=-5.4e-004$ and, therefore, $SEC_{1,2}$ are stable cycles. For $H_{1,2}$, $MLE=0.0057>0$, i.e. these attractors are chaotic.

\subsection{Coexisting self-excited cycles and hidden cycles}\label{particula}

\underline{$a=0.0509$}. The zoomed region $D$ defined for $a\in[0.0503,0.0515]$ (Fig. \ref{fig2} (b)) contains merged periodic and chaotic thin windows. Within these interesting windows, one can identify the coexistence of several types of possible periodic attractors. The uncertainty relates to thickness of these windows, to the six decimals of some coefficients of the system and also to the errors introduced by the numerical integration. Also, the regularity seems to be influenced by the close proximity of the periodic windows with thin chaotic bands.

 To increase the precision of the integration, in the build-in Matlab procedure ode45, the option $options = odeset('RelTol',1e-8)$ is used (the default relative tolerance $RelTol$ of ode45 integrator, $0.001$, cannot separate, in a reasonable time integration interval $T_{max}$, the chaotic transients from the stable cycles).

For $a=0.0509$, because of the mentioned utilized $options$, the equilibria $X_{1,2}^*$ become $X_{1,2}(\pm2.9032,2.2064,\mp2.9032)$, and there exist two pairs of stable cycles:  \emph{self-excited cycles} $SEC_1$ and $SEC_2$ (red and black in Fig.\ref{fig7} (a)) and \emph{hidden cycles} $HC_1$ and $HC_{2}$ (blue and light-brown). The presumed stability can be deduced from the time series in Fig. \ref{fig7} (b), while the type of attractors can be deduced as before from the sections through equilibria $X_0^*$, the lattice $B_0$ through $X_0^*$ (Fig. \ref{fig7} (c)) and $B_1$ through $X_1^*$ (Fig. \ref{fig7} (d)).

Again, a small perturbation of the bifurcation parameter with only $\Delta a=9e-4$ (from $a=0.05$ to $a=0.0509$), generates a modification in the attraction basins containing $X_1^*$ (compare Fig. \ref{fig7} (d) with Fig. \ref{fig6} (d)), so that the hiddenness property appeared again.

For both pairs of cycles, $MLE=-4.25e-4$.

\section{FO variant of the system \eqref{3}}
Probably the most interesting characteristic of this system is the coexistence of attractors in the space of the fractional-order systems in the sense that for some fractional order $q$ and fixed parameter $a$, there exist two symmetric pairs of coexisting chaotic and stable cycles attractors.

The commensurate FO variant of the system is

\begin{equation}\label{33}
\begin{array}
[l]{l}
D_*^q x_1=ax_2+0.01x_1(1-x_2^2),\\
D_*^q x_2=0.031847(x_1+x_3),\\
D_*^q x_3=    0.19x_1-0.25x_2,
\end{array}
\end{equation}
where $D_*^q$ denotes the Caputo differential operator of order $q\in(0,1)$ with starting point 0 (see e.g. \cite{10,21})

\[
D_*^q=\frac{1}{\Gamma(\lceil q\rceil-q)} \int_0^t(t-\tau)^{\lceil q\rceil-q-1}D^{\lceil q\rceil}x(\tau)d \tau,
\]
where $\lceil \cdot \rceil$ denotes the ceiling function that rounds up to the next integer and $D^{\lceil q\rceil}$ represents the standard differential operator of order $\lceil q \rceil\in\mathbb{N}$. Due to the use of Caputo's operator, the initial conditions can be taken as for the integer-order case.

\begin{lemma} For $a\in[0,0.2]$, equilibria $X_0^*$ and $X_{1,2}^*$ are unstable.
\end{lemma}
\begin{proof}
In order to study the stability of equilibria $X_0^*$ and $X_{1,2}^*$ of the system \eqref{33}, denote $\alpha_{min}=min\{|\alpha _i|\}$, for $i=0,1,2$, where $\alpha_i$ are the arguments of the eigenvalues. The stability theorem of equilibria of a FO system can be stated in the following practical form

\begin{theorem}\cite{tava,zece}
An equilibrium point $X^*$ is asymptotically stable if and only if the instability measure
\[
\iota=q-2\frac{\alpha_{min}}{\pi}
\]
is strict negative.
\end{theorem}

If $\iota \leq 0$ and the critical eigenvalues satisfying $\iota=0$  have geometric multiplicity one, $X^*$ is stable.\footnote{The geometric multiplicity represents the dimension of the eigenspace of corresponding eigenvalues.}

The graphs of $\iota$ for all equilibria, with $a\in[0,0.2]$, are presented in Fig. \ref{fig77}. As can be seen $\iota>0$ and, therefore, equilibria are unstable.
\end{proof}

Consider next $a=0.05$. The bifurcation diagram with respect the fractional order $q$ is presented in Fig. \ref{fig8} (a). As can be seen, at the right of the bifurcation diagram, there exists a FO coexisting window where two symmetric pairs of attractors coexist.

\begin{remark}
Generally for FO systems, attractors coexisting phenomenon is searched for a fixed fractional order $q$ in the parameter bifurcation space, by searching different initial conditions for some fixed parameter value. For this system of FO the coexistence of attractors can be found also in the fractional order $q$ space, by searching for initial conditions for a fixed parameter bifurcation.
\end{remark}

Note that, as for continuous integer-order systems and contrary to FO difference systems, where the chaos strongness decreases for increasing $q$ (see e.g. \cite{11}), for this system of FO, the chaos strongness increases with $q$ tending to 1.

Because the inherent time history of the ABM method, a deep numerical analysis to detect hidden attractors is tedious and time consuming. Therefore, the quality of the obtained attractors is determined by several empirical tests.

\noindent \underline{$q=0.9995$}, a value within the coexisting window (Fig. \ref{fig8} (a)). Equilibria $X_{1,2}^*$ are the same as for the IO case, for $a=0.05$, namely $X_1^*(2.8828,2.1909,-2.8828)$. The two symmetric pairs of FO-attractors are presented in Figs. \ref{fig8} (b), (c) and Figs. \ref{fig8} (d), (e) (phase portraits and time series, respectively). As can be seen, each pair of FO-attractors is composed by one \emph{self-excited cycle}, $SEC_1$ or $SEC_2$ (red and black respectively) and one \emph{hidden chaotic attractor} $H_{1}$ or $H_{2}$ (blue and light-brown, respectively).

Now, the cycles have $MLE=5e-4$ while for the chaotic attractors $MLE=0.0056$.

Therefore, for a given parameter value $a$, the FO system \eqref{33} admits at least one FO coexisting window in the space of fractional order $q$, with two symmetric pairs of coexisting attractors. Other potential FO-coexisting windows, not examined here, are at the bifurcation points, e.q. at about $a=0.985$.

This important characteristic indicates that for some given fractional order $q$ and parameter $a$, there could be two different systems.

\section{Numerical settings}

Choosing a suitable level of precision for computations is a real challenging task.
The numerical integrator for \eqref{3} in the IO case is the matlab ode45 with implicit absolute and relative error tolerances. For the particular case of $a=0.0509$ (Subsection \ref{particula}), errors have been reduced by using $options = odeset('RelTol',1e-8)$.

The maximal time used interval is $T_{max}=5000-8000$, depending on the lengths of discarded transients. Due to inherent numerical errors, larger values for $T_{max}$ could lead to invalid results.

To integrate numerically the system \eqref{33} of FO, the ABM method \cite{21} is utilized.

The algorithms used to determine the finite time MLEs are the Benettin-Wolf algorithm for the IO case, and the algorithm presented in \cite{dn} for the FO case.

The attractors data are presented in Table \ref{pill}. Beside the initial conditions presented in Table \ref{pill}, points $(\pm1,\pm1,\pm1)$ and $(\pm1,\pm1,\mp1)$ can be used as initial conditions.
\begin{table}
\small
$\begin{array}{ccccc}
& a & x_{0} & \text{Attractor} & \text{Figure} \\
\hline
& 0.052 &
\begin{array}{c}
(\pm 0.0720,0,0),~(\pm 1,\pm 1,\pm 1) \\
(\pm 0.250,0,0),~(\pm 1,\pm 1,\mp 1)%
\end{array}
&
\begin{array}{c}
H_{1,2} \\
SECH_{1,2}%
\end{array}
& \text{Figs. 3,4} \\
\addlinespace[.2cm]
IO & 0.05 &
\begin{array}{c}
(\pm 0.150,0,0),(\pm 1,\pm 1,\pm 1) \\
(\pm 0.350,0,0),~(\pm 1,\pm 1,\mp 1)%
\end{array}
&
\begin{array}{c}
SEC_{1,2} \\
SECH_{1,2}%
\end{array}
& \text{Fig. } 6 \\
\addlinespace[.2cm]
& 0.0509 &
\begin{array}{c}
(\pm 0.150,0,0),~(\pm 1,\pm 1,\pm 1) \\
(\pm 0.280,0,0)~(\pm 1,\pm 1,\mp 1)%
\end{array}
&
\begin{array}{c}
SEC_{1,2} \\
HC_{1,2}%
\end{array}
& \text{Fig. } 7 \\
\addlinespace[.1cm]
\hline
FO & 0.05 &
\begin{array}{c}
(\pm 1,\pm 1,\pm 1) \\
(\pm 1,\pm ,\mp 1)%
\end{array}
&
\begin{array}{c}
SEC_{1,2} \\
H_{1,2}%
\end{array}
& \text{Fig. } 9\\
\hline
\end{array}$
\caption{Attractors data}\label{pill}
\end{table}

Note that, because of the relatively large number of decimals of some system coefficients and also due the inherent numerical errors, the 4-digits coordinates of $X_{1,2}^*$ (as approximations of the exact $X_{1,2}^*$) could also be taken as initial conditions, for cases with neighborhoods of $X_{1,2}^*$.

\section*{Conclusion and Discussion}
In this paper several type of coexisting attractors of an economic system of IO and FO are presented. Considering one of the several parameters as bifurcation parameter, the bifurcation diagram revealed windows where hidden and self-excited attractors coexist.

First, it is shown numerically that equilibria are unstable for both IO and FO cases, after which the exploration within neighborhoods of these points is investigated. The experiments, realized with the matlab routine ode45 for the IO case and with the ABM method for the FO case, show that the considered attractors fit the definition of hidden or self-excited attractors.

The characteristics of chaotic or regular attractors are revealed, after transients removed, by phase portraits, time series and maximal Lyapunov exponents.

For the bifurcation diagram in Fig. \ref{fig1}, the four used initial conditions are $(\pm1,\pm1,\pm1)$ and $\pm1,\pm1,\mp1)$, respectively. These initial conditions are proved to be valid, besides those mentioned in the text, for all considered cases (Table \ref{pill}). Nevertheless, to note that generally this is a way to draw bifurcation diagrams are made for fixed initial conditions, it is not a rigorous method, since the attraction basins be modified while the bifurcation parameter varies. Thus some searched attraction basins, such as hidden attractors, might modify once the parameter is varied and, in this way, hidden attractors can ``hide'' from this searching method. Indeed, one can consider as initial approximations of unstable equilibria (if any), whose expressions depend on parameter, which will lead certainly to all self-excited attractors (if any). But, the chance to identify in this case hidden attractors diminishes significantly. Moreover, sudden changes in the bifurcations branches may indicate bifurcations when some branches may become unstable, but also they might be due to the mentioned dependence of the attraction basins with the parameter modification, when the bifurcation diagram is incomplete. For example, for this system of IO, for values of $a$ close to $a=0.05$, with relative small perturbations of $a$, of order of $1e-3$ or even of $1e-4$ (case $a=0.05$ and $a=0.0509$), the attraction basins change their properties related to hiddenness and self-excited.

A coexisting window in the space of fractional order indicates that analyzing a FO system by following only the bifurcation in the parameter space, without some analysis of the bifurcation diagram in the fractional order space, could be insufficient to identify the entire spectrum of system dynamics.

An important problem, for this system, is a deeper analysis of his FO variant.

The coexistence of attractors in the case of this system is influenced by the presence of the saddles on the separatrix of attraction basins. For all these cases, when the unstable manifolds of the saddle could separate the attraction basins (see the cases in Fig. \ref{fig4} (c), Figs. \ref{fig6} (c) and (e) and Fig. \ref{fig7} (c)), this problem could be analyzed further as a possible sufficient criterion for non-existence of hidden attractors.

Because the variable $x_3$ represents the foreign capital inflow \cite{doi}, the presence of hidden attractors could represent an important phenomenon from the economic point of view. In fact under some circumstances such as initial conditions, the dynamics of the system could cover some (hidden) behavior. Also, counterintuitively, due to the existence of hidden attractors, increasing (but also decreasing) the $x_3$ variable is not necessarily related with the system economical stability.

\newpage{\pagestyle{empty}\cleardoublepage}

\newpage{\pagestyle{empty}\cleardoublepage}

\newpage{\pagestyle{empty}\cleardoublepage}

\begin{figure}
\begin{center}
\includegraphics[scale=0.5]{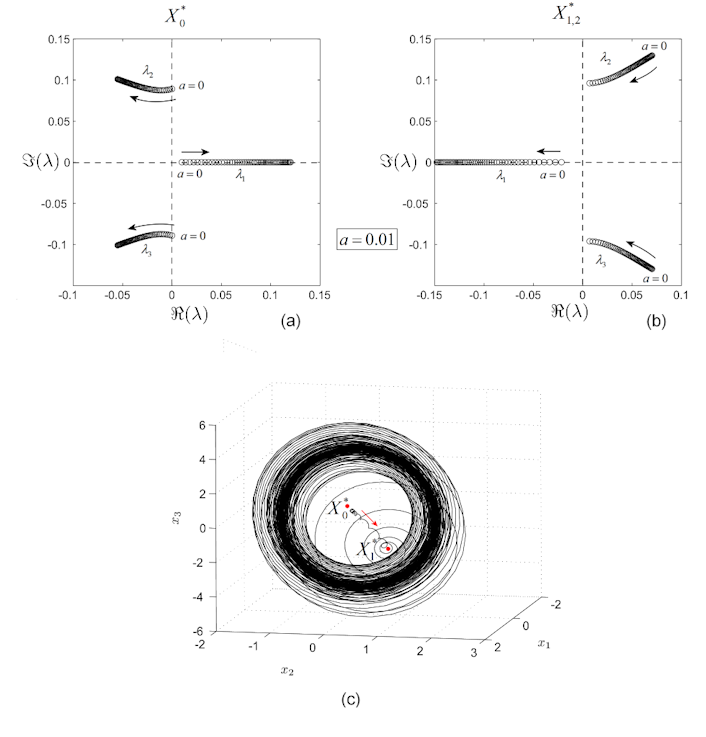}
\caption{(a) Locus of eigenvalues of $X_0^*$; (b) Locus of eigenvalues of $X_{1,2}^*$; (c) Heteroclinic connection between $X_0^*$ and $X_1^*$ for $a=0.001$.}
\label{fig1}
\end{center}
\end{figure}

\begin{figure}
\begin{center}
\includegraphics[scale=0.75]{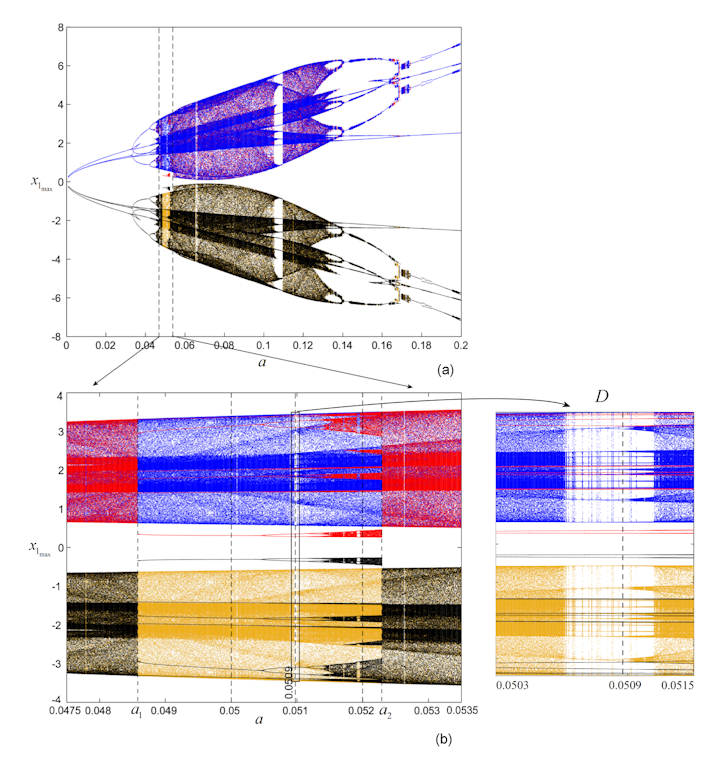}
\caption{(a) Bifurcation diagram of $x_1$ of the system \eqref{3} of IO for $a\in[0,2]$; (b) Two successive zooms. }
\label{fig2}
\end{center}
\end{figure}

\begin{figure}
\begin{center}
\includegraphics[scale=0.75]{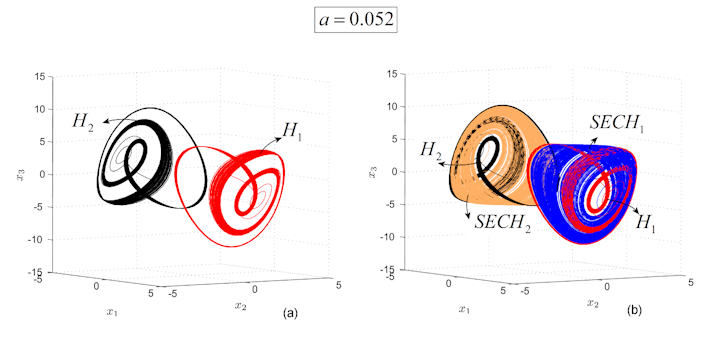}
\caption{Coexisting hidden chaotic attractors and self-excited attractors of the IO system, for $a=0.052$; (a) Hidden chaotic attractors $H_1$ (red plot) and $H_2$ (black plot) for initial conditions $x_0=(\mp0.072,0,0)$ respectively; (b) Self-excited chaotic attractors $SECH_1$ (blue) and $SECH_2$ (light brown) for $x_0=(\mp0.014,0,0)$ respectively.}
\label{fig3}
\end{center}
\end{figure}

\begin{figure}
\begin{center}
\includegraphics[scale=0.5]{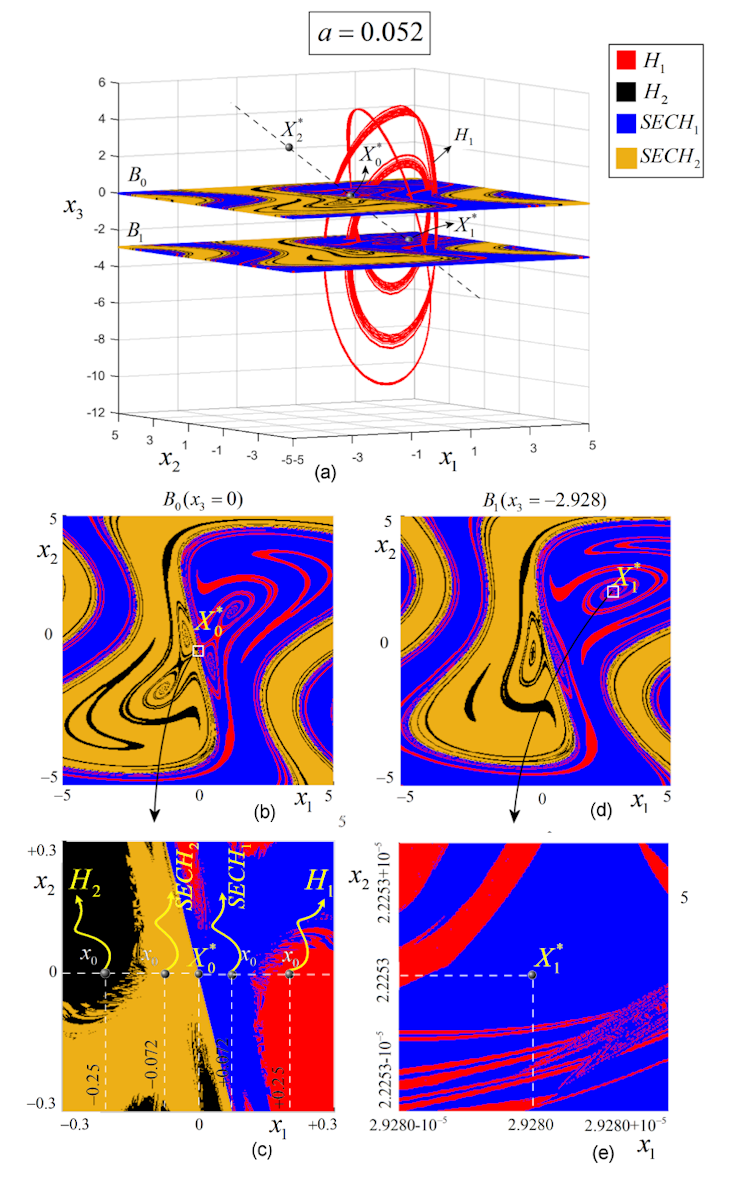}
\caption{Attraction basins of the chaotic attractors of the IO system, for $a=0.052$. Red and black correspond to the hidden attractors ($H_1$ and $H_2$, respectively), while blue and light-brown correspond to self-excited attractors ($SECH_1$ and $SECH_2$, respectively); (a) Three-dimensional view of two planar horizontal sections through $H_1$ and equilibrium $X_0^*$ (plane $B_0$, $x_3=0$) and also through equilibrium $X_1^*$ (plane $B_1$, $x_3=-2.9280$). Both planes are designed as lattices: $B_{0,1}=[-5,5]\times [-5,5]$; (b) View of $B_0$; (c) Rectangular zoomed region centered at $X_0^*$; (d) View of $B_1$; Rectangular zoomed region centered at $X_1^*$; Points $x_0$ represents initial conditions, while scrolled arrows show to which attractors initial conditions tend.}
\label{fig4}
\end{center}
\end{figure}

\begin{figure}
\begin{center}
\includegraphics[scale=0.5]{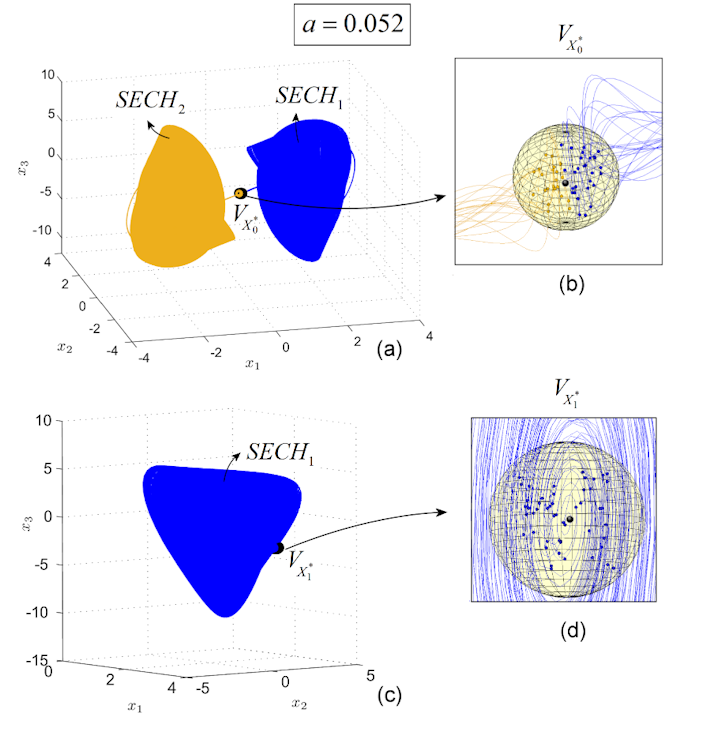}
\caption{Three-dimensional spherical neighborhoods of equilibrium $X_0^*$ and $X_1^*$ of the IO system, for $a=0.052$; (a) Neighborhood $V_{X_0^*}$ of $X_0^*$; (b) Zoom of the neighborhood $V_{X_0^*}$; (c) Neighborhood $V_{X_1^*}$ of $X_1^*$; (d) Zoom of the neighborhood $V_{X_1^*}$. }
\label{fig5}
\end{center}
\end{figure}

\begin{figure}
\begin{center}
\includegraphics[scale=0.65]{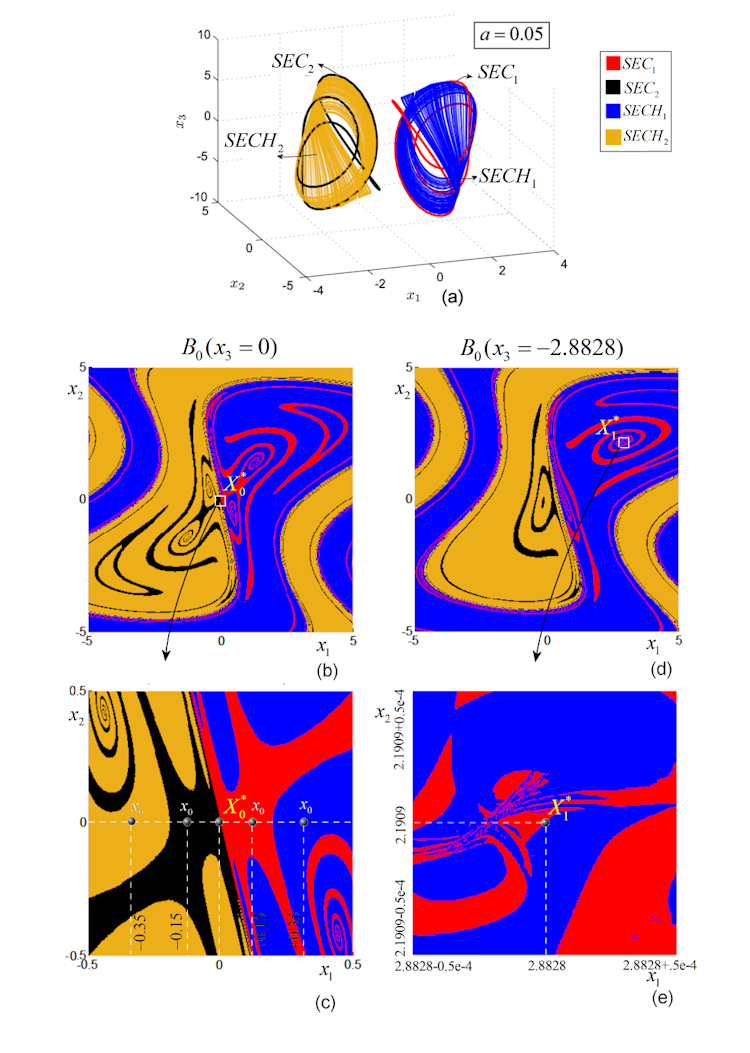}
\caption{Self-excited cycles $SEC_{1,2}$ and self-excited chaotic attractors $SECH_{1,2}$ of the IO system, for $a=0.05$; (a) Three-dimensional phase plot; (b) View of lattice $B_0$ of the planar section with plane $x_3=0$ containing equilibrium $X_0^*$; (c) Zoom of a neighborhood of $X_0^*$; (d) View of lattice $B_1$ of the planar section with plane $x_3=-2.8828$ containing equilibrium $X_1^*$; (e) Zoom of a neighborhood of $X_1^*$; Points $x_0$ represents initial conditions.}
\label{fig6}
\end{center}
\end{figure}

\begin{figure}
\begin{center}
\includegraphics[scale=0.65]{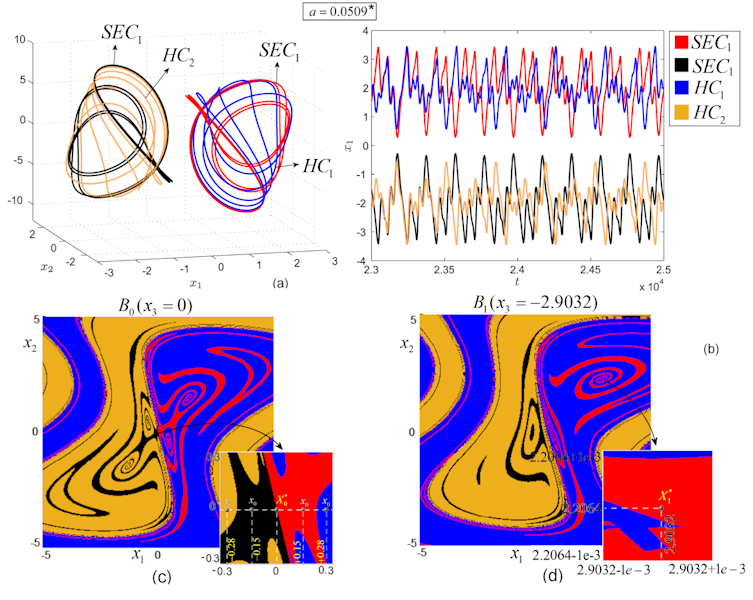}
\caption{Self-excited cycles $SEC_{1,2}$ (red and blue plot) and hidden cycles $HC_{1,2}$ (black and light-brown plot) of the IO system, for $a=0.0509$; The star indicates that the numerical method is improved in this case; (a) Three-dimensional phase plot; (b) Time series; (c) View of lattice $B_0$ of the planar section with plane $x_3=0$ containing equilibrium $X_0^*$ and a zoomed region of $X_0^*$; (d) View of lattice $B_1$ of the planar section with plane $x_3=-2.9032$ containing equilibrium $X_1^*$ and a zoomed region of $X_1^*$.}
\label{fig7}
\end{center}
\end{figure}

\begin{figure}
\begin{center}
\includegraphics[scale=0.5]{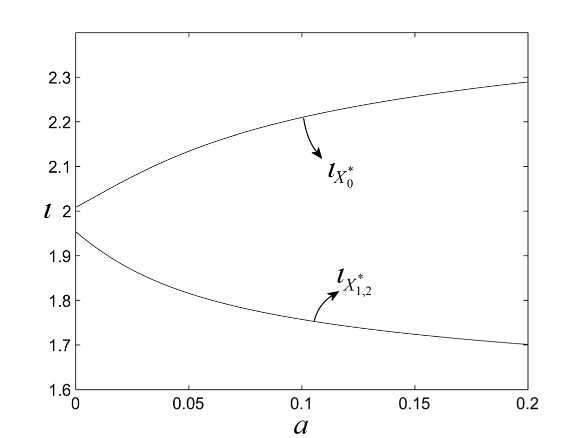}
\caption{The instability measure $\iota$ of equilibria $X_0$ and $X_{1,2}^*$, for the FO system with $q=0.9995$, as function of $a$.}
\label{fig77}
\end{center}
\end{figure}

\begin{figure}
\begin{center}
\includegraphics[scale=0.65]{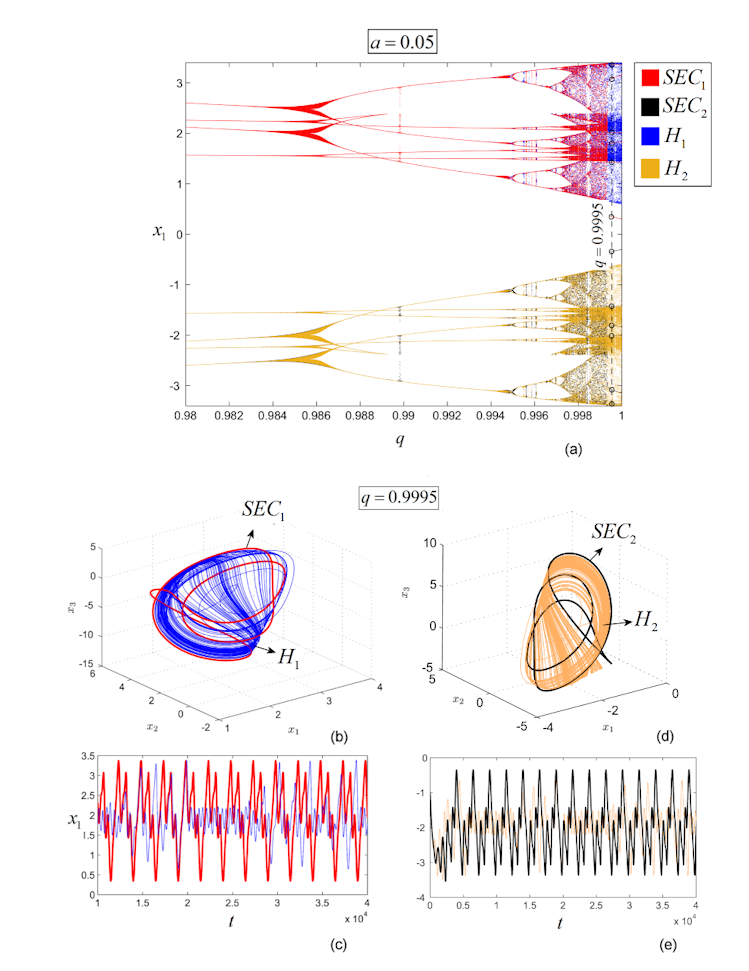}
\caption{The FO system for $a=0.05$ and $q=0.9995$; (a) Bifurcation diagram for $q\in[0.98,1)$; (b) Self-excited stable cycle $SEC_{1}$ and hidden chaotic attractor $H_1$; (c) Self-excited stable cycle $SEC_{2}$ and hidden chaotic attractor $H_2$; (d)Time series for attractors $SEC_1$ and $H_1$; (e) Time series for attractors $SEC_2$ and $H_2$.}
\label{fig8}
\end{center}
\end{figure}

\end{document}